\documentclass[copyright,creativecommons]{eptcs}
\providecommand{\event}{WRS 2009} 
\providecommand{\volume}{??}
\providecommand{\anno}{2009}
\providecommand{\firstpage}{1}
\providecommand{\eid}{??}
\usepackage{breakurl}        
\usepackage{amsthm}   
\usepackage[pdftex]{graphicx}

\newcommand{\letraEstado} [1] {{\small\texttt{#1}}}
\newcommand{\letraFunc} [1] {{\small{#1}}}
\newtheorem{lem}{Lemma} 
\newtheorem{thm}{Theorem}

\title{Specification of Products and Product Lines}
\author{Ariel Gonzalez
\institute{Universidad Nacional de Rio Cuarto\\ Rio Cuarto, Argentina}
\email{agonzalez@dc.exa.unrc.edu.ar}
\and
 Carlos Luna
\institute{Universidad ORT\\
Montevideo, Uruguay}
\email{luna@ort.edu.uy}
}
\def\titlerunning{Specification of Products and Product Lines}
\def\authorrunning{A.\ Gonzalez \& C.\ Luna}

\begin{document}
\maketitle

\begin{abstract}
The study of variability in software development has become increasingly important in recent years. A common mechanism to represent the variability in a product line is by means of feature models. However, the relationship between these models and UML design models is not straightforward. UML statecharts are extended introducing variability in their main components, so that the behavior of product lines can be specified. The contribution of this work is the proposal of a rule-based approach that defines a transformation strategy from extended statecharts to concrete UML statecharts. This is accomplished via the use of feature models, in order to describe the common and variant components, in such a way that, starting from different feature configurations and applying the rule-based method, concrete state machines corresponding to different products of a line can be obtained.
\end{abstract}

\section{Introduction}

A product line (PL), also called system family, is a set of software systems sharing a common, managed set of features that satisfy the specific needs of a particular market segment or mission and that are developed from a common set of core assets in a prescribed way \cite{CN02,Gom04,KLD02}.

To develop a system family as opposed to developing a set of isolated systems has important advantages. One of these is the possibility of building a kernel that includes common features, from which desired products can be easily built as extensions. By adding distinguishing characteristics (variability) to such a kernel, different products can be obtained \cite{KLD02,vdML02,HP03,ZDD06,Cla01,CA05}. For example, today we observe in the market a significant number of different types of mobile phones (MPs) that share a core of basic features and differ in other more specific characteristics: availability of a digital camera, internet access, mp3 player, among others.

The UML language \cite{OMG04} provides a graphical notation and has become the standard for modeling different aspects of software systems. Statecharts and interaction diagrams are part of the set of tools that UML provide so that the system behavior can be specified, which are specially suitable for the software design phase. Statecharts are used to specify the behavior of the instances of a class (intra-component behaviour), and therefore constitute an appropriate mechanism for describing the behavior of certain problems by means of a graphical representation. In the latest version of UML 2.0, the statecharts do not offer operators and/or sublanguages for specifying system families.
 
In this paper we propose an extension of UML statecharts for modeling PLs. We used feature models so that both common and variant functionality of a system family can be described \cite{CE00,CA05}, and we incorporate variability in the essential components of the statecharts, in such a way that, starting from different configurations of a feature model, concrete statecharts corresponding to different products of a PL can be generated applying a rule-based method. The approach defines the transformation strategy from extended statecharts to concrete UML statecharts.

The rest of the work is organized as follows. In section 2 and 3 we briefly introduce statecharts and feature models, respectively. Section 4 presents an extension of statecharts with variant elements, which together with the use of feature models allow specifying PLs. In section 5 we detail the mechanisms for obtaining products of a PL from distinct configurations of a feature model, via the use of rules and a transformation strategy. Related work is discussed in section 6 and finally, we conclude and discuss possible further work in section 7. We exemplify the proposed work by developing part of a case study based on mobile phone technology. A preliminary version of this work is \cite{SCCC08}. As opposed to \cite{SCCC08}, this paper presents a rule-based approach. We define an application strategy for the rules in a proper manner, in such a way that inconsistencies are avoided in the statechart obtained. The rules are organized in a sequence of rule sets, in which each rule set can be considered as a layer. Within a rule set, the rules may be applied in a non-deterministic order \cite{JK06}. We also formalized and added rules that in \cite{SCCC08} are omitted or only described informally. 

\section{Statecharts}
UML StateCharts (SCs) constitute a well-known specification language for modeling the dynamic system behavior. SCs were introduced by D. Harel \cite{Har87} and later incorporated in different versions of the UML with some variations. In this section, we present definitions of SCs based on \cite{vdB02}. For additional details, the reader is referred to \cite{SCCC08,vdB02}.

SCs consist essentially of states and transitions between states. The main feature of SCs is that states can be refined, defining in this way a state hierarchy. A state decomposition can be sequential or parallel. In the first case, a state is decomposed into an automata (Or-state). The second case is supported by a complex statechart composed of several active sub-statecharts (And-state), running simultaneously. 

Let $S, TR, \Pi$ and $A (\Pi \subseteq A)$ be countable sets of state names, transition names, events and actions of a SC, respectively. Also, let us define $s \in S$ as either a basic term of the form $s = [n]$ (Simple-state), as a term Or of the form $s = [n, (s_1,...,s_k), l, T]$ (Or-state), or as a term And of the form  $s = [n, (s_1,...,s_k)]$ (And-state), where $name(s)=_{def} n$ is the name of the state $s$. Here $(s_1,...,s_k)$ are the subterms (substates) of $s$, also denoted by $sub\_est(s)=_{def} (s_1,...,s_k)$. Likewise, $inicial(s)=_{def} s_1$ is the initial state of $s$, $T \subseteq TR$ is the set of internal transitions of $s$, and $l$ the active state index of $s$. A transition is represented as a tuple $t = (t', s_o, e, c, \alpha, s_d, ht)$, where $name(t)=_{def} t'$ is the transition name, $source(t) =_{def} s_o$ and $target(t) =_{def} s_d$  are called source and target of $t$, respectively, $ev(t) =_{def} e$ the trigger event, $cond(t) =_{def} c$ the trigger condition, and $acc(t) =_{def} \alpha$ is the sequence of actions that are carried out when a transition is triggered. In addition, $hist(t) =_{def} ht$ is the history type of the target state of $t$ \cite{vdB02}. The graphical notation used in the transitions is $t: e,c /\alpha$. 

\section{Feature Models}
Feature Models (FMs) are used to describe properties or functionalities of a domain. A functionality is a distinctive characteristic of a product or object, and depending of the context it may refer to, it is a requirement or component inside an architecture, and even code pieces, among others. FMs allow us to describe both commonalities and differences of all products of a PL and to establish relationships between common and variant features of the line. There are multiple notations for describing FMs. In this work we will use the proposal of Czarnecki \cite{CE00}.

A tree structure instance is a \emph{FM configuration (FMConf)} that describes the model and that respects the semantics of their relations. That is, a FM allows one to identify common and variant features between products of a PL, while a FM configuration characterizes the functionalities of a specific product. Formally, the concepts of FMs are defined as follows:

\textbf{Definition 1.}
A \emph{FM} is defined as a tree structure represented by a tuple ($Funcs, f_0, Mand, Opt, Alt,\\ Or$-$rel$), where $Funcs$ is a set of functionalities of a domain (nodes of the tree), $f_0 \in Funcs$ is the root functionality of the tree and, $Mand, Opt, Alt, Or$-$rel \subseteq Funcs \times ( \gamma (Funcs)- \{\emptyset\}  )$ the mandatory, optional, alternative and disjunct relations of the model, respectively. If (f,sf)$\in Mand \cup Opt$, \#sf = 1.

\textbf{Definition 2.}
A \emph{FM configuration} corresponding to a \emph{FM} ($Funcs,\:f_0,\:Mand,\:Opt,\:Alt,\:Or$-$rel$) is a tree $(F,\:R)$ where $F$ is the set of nodes and $R$ the set of edges; $F \subseteq Funcs$ and $R \subseteq$  \{(f,sf)$ \in F \times$($\gamma$($F$)$-$\{$\emptyset$\}) $| \; \exists$ sf' $\in \gamma$(Funcs): sf $\subseteq $ sf' $\wedge $ (f,sf') $\in Mand \cup Opt \cup Alt \cup Or$-$rel$\}. Moreover, the following conditions must be fulfilled by $(F,\:R)$: (1) $f_0 \in F$; (2) for every (f,sf) $\in Mand$: if $f \in F$ then (f,sf) $\in R$; (3) if (f,sf) $\in Alt$ $\wedge \:f \in F$  then $\exists !$ sf'$\in \gamma$($F$): sf'$\subseteq$sf $\wedge$ (f,sf') $\in R$ $\wedge$ \#sf'=1; (4) if (f, sf) $\in Or$-$rel$ $\wedge$ $f \in F$ then $\exists !$ sf'$\in \gamma$($F$): sf'$\subseteq$sf $\wedge$ (f,sf') $\in R \:\wedge$ \#sf'$\geq$1.

\textbf{Definition 3.}
The \emph{kernel N} of a FM ($Funcs,\:f_0,\:Mand,\:Opt,\: Alt,\: Or$-$rel$) is the set of functionalities, which are present in all configurations, inductively characterized by the following rules: (1) $f_0 \in N$; (2) if $f_1 \in N$ $\wedge$ $(f1, \{f2\}) \in Mand$ then $f_2 \in N$.

\section{SCs with Variabilities}
In this section we extend the SCs with \emph{optional (variant) elements} and later on we establish the binding between these elements with functionalities of a FM, in order to model the behavior of a PL. We will call our proposed extended machines StateCharts$^*$ (SCs$^*$).

\subsection{Graphical Representation of a SC$^*$}
The representation of the optional elements that extend the system kernel in a SC$^*$ are depicted in figure~\ref{fig:EstadoyTransicionOpcional}. We use dashed lines to graphically denote both optional states as well as transitions.
\begin{figure}[!h]
\begin{center}
 \includegraphics [width=70mm, height=10mm] {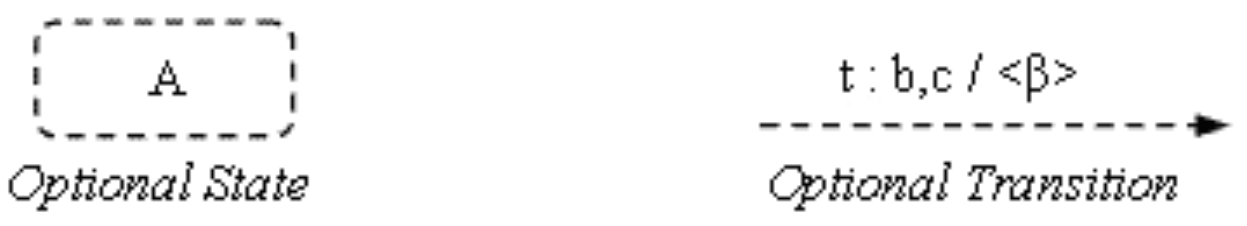}
\end{center}
\caption{Optional state and Optional transition.}
\label{fig:EstadoyTransicionOpcional}
\end{figure}

\subsection{Abstract Syntax of a SC$^*$}
Let $S^*,\; TR^*, \Pi^*$ and $A^*$ be set of states, transitions, events and actions of a SC$^*$, respectively. Now the terms that define a state have an additional component $S_{op} \in \{optional, no\_optional\}$ that we will call \emph{StateType(s)}, which indicates whether the state \emph{s} is optional or not. Similarly, we add component $t_{op} \in \{optional, no\_optional\}$ to the transitions, and we denote it by \emph{TransType(t)}. We also define the following sets of SCs$^*$ optional elements: $SOp \subseteq S^*$, $TOp \subseteq TR^*$ and $VarElem = SOp \cup TOp$.

We will refer to states directly by their names, when these are unique for every state in all the SC$^*$; otherwise, we will use the dot (.) as separator between state and substate names. A transition name is built by the trigger event name followed by source and target state names, respectively.

\subsection{Case study: MPs}
\label{Case study: MPs}
We considered here a family of MPs which share some functionalities, such as, for example, the capacity of reproducing monophonic sounds and vibration. Optionally, we could incorporate into the kernel of functionalities the capacity to make calls by means of quick-marked, to write text messages, to administer multimedia contents, and combinations of these, such as messages with multimedia content (images, polyphonic sounds and videos).
In order to exhibit an example in the development of this article, we formulate in figure \ref{fig:fm-del-ejemplo} a FM, using the notation proposed by Czarnecki, that relates the involved functionalities in the partially described MP of figure \ref{fig:SC*deunTM}.

\subsection{Relation between FMs and SCs$^*$}
\label{Relation between FMs and SCs*}
FMs and SCs$^*$ are complementary. Both model different aspects of a system and in our proposal, will not be treated independently, since SC$^*$ elements model behaviors of present functionalities in the FM. In general, a functionality is described by more than one SC$^*$ element. Due to this, we define a function \emph{Imp}, which represents the association between the SC$^*$ variant elements and the functionalities of the FM. This way we establish what variant elements of the SC$^*$ implement the characteristics of the system described in the FM. Given a FM ($Funcs,\: f_0,\: Mand,\: Opt,\: Alt,\: Or$-$rel$) and a SC$^*$ ($S^*,\; TR^*,\; \Pi^*$ and $A^*$), the type of function \emph{Imp} is as follows: $Imp: Funcs \rightarrow  Set(VarElem)$, where \emph{VarElem} is the set $SOp \cup TOp$, $SOp \subseteq S^*$ and $TOp \subseteq TR^*$.  

\begin{figure}[!h]
\begin{center}
\includegraphics [width=150mm, height=155mm]{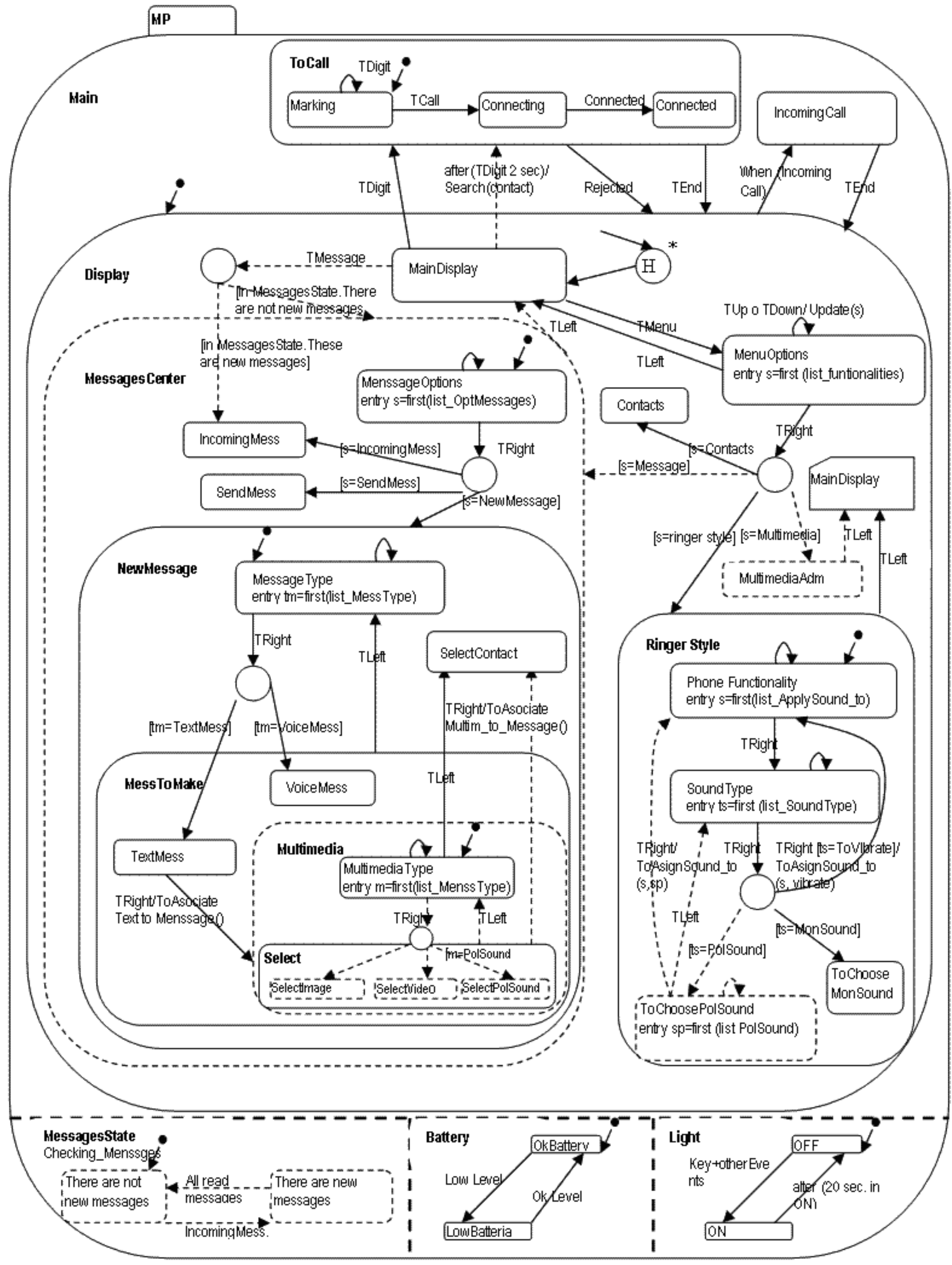}
\end{center}
\caption{SC$^*$ of MPs.}
\label{fig:SC*deunTM}
\end{figure}

\begin{figure}[!h]
\begin{center}
\includegraphics [width=115mm, height=25mm]{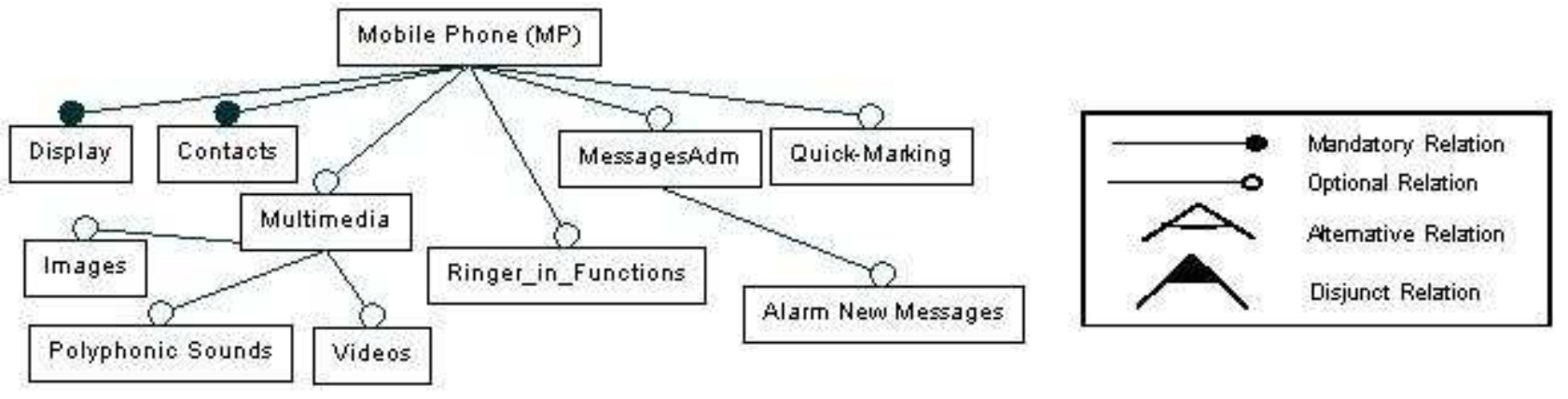}
\end{center}
\caption{FM of SC$^*$ of figure  \ref{fig:SC*deunTM}.}
\label{fig:fm-del-ejemplo}
\end{figure}

Taking into account that the mandatory functionalities are always present in all products of the line, it is not necessary to define the SC$^*$ syntactic elements that these implement. However, it is necessary to do it for those functionalities that cannot belong to FM configuration. \emph{Imp} will be then a \emph{partial function} defined on FM elements which do not belong to the \emph{kernel}. Therefore, the behavior of a PL is defined by a FM, a SC$^*$, and a function of implementation that binds them. 

\emph{\textbf{Example 1}}. Taking the FM of figure \ref{fig:fm-del-ejemplo} and the SC$^*$ of case study of figure \ref{fig:SC*deunTM}, we define SC$^*$ elements that implements functionalities of MP as follows:\\
$Imp$(\letraFunc{Polyphonic Sounds}) = \letraEstado{\{SelectPolSound, TRightMultimediaType-SelectPolSound,
 ToChoosePol\\Sound, TRightSoundType, ToChoosePolSound, TLeftToChoosePolSound-SoundType,
 TRightToChoo\\sePolSound-PhoneFuncionality, ...\}};\hspace{0.3cm}
$Imp$(\letraFunc{Multimedia}) = \letraEstado{\{Multimedia, AdmMultimedia, TRight\\Multimedia.Selecct-SelectContact\}}
 $\cup$  $Imp$(\letraFunc{Images}) $\cup$ 
 $Imp$(\letraFunc{Polyphonic Sounds}) $\cup$ $Imp$(\letraFunc{Videos});\\
 $Imp$(\letraFunc{MessagesAdm}) = \letraEstado{\{MessagesCenter,  TMessage-MainDisplay-IncomingMess, TMessage-MainDis\\play-MessagesCenter, \hspace{0.55mm} TLeft-MessagesCenter-MainDisplay,\hspace{0.55mm}
 TRight-OptionsMenu-MessagesCen\\ter\}} $\cup$ $Imp$(\letraFunc{Alarm New Messages});\hspace{0.3cm}
 $Imp$(\letraFunc{Alarm New Messages}) = \letraEstado{\{MessagesState, TMessageMainDis\\play-IncomingMess\}}. 

\section{Instantiation of StateCharts with Variabilities}
\label{Instantiation of StateCharts with Variabilities}
A FM configuration defines a product given a set of selected characteristics. Given a FM configuration and the SC$^*$ corresponding to the FM (linked via a function \emph{Imp}), we define an \emph{instantiation function} that returns a SC, which specifies the defined product behavior, $Inst: SC^* \times FMConf \rightarrow SC$.

We eliminate of the SC$^*$ both states and transitions which implement functionalities not present in FM configuration, via the use of the function \emph{Imp} defined in the previous section. The direct elimination of states as well as transitions of the SC$^*$ is not trivial. The suppression of SC$^*$ components without establishing a control can return inconsistent results, such as, for example, unreachable states or transitions without target. A control and rebuild mechanism of SCs starting from a SC$^*$ is defined in such a way that a concrete product is obtained. 

In section \ref{Cases and Rebuilding Rules} we present the cases and rules of rebuilding that constitute the base of the instantiation method which we included in section \ref{Instantiation Method}. Later, in section \ref{Instantiation of the Case study: MPs} we analyze our case study: MPs.

\subsection{Cases and Rebuilding Rules}
\label{Cases and Rebuilding Rules}
\hspace{0.58cm}Case 1. \emph{\textbf{When a state is deleted}}

Case 1.1.  \emph{\textbf{When a simple state is deleted}}

If a simple state $s = [E]$ is deleted, then their entry and exit optional transitions are deleted, while the mandatory transitions are composed using the following rebuilding method.

Let $E \in SOp$ be the state to eliminate, $A_1, ...,\:A_n$ predecessor states of \emph{E} (i.e., states from which there are non-optional transitions with target \emph{E}: $t_{AE\_1},\: ...,\: t_{AE\_n}$), and $S_1,\: ...,\: S_m$ successor states of \emph{E} (i.e., target states of non-optional transitions with source \emph{E}: $t_{ES\_1},\: ...,\: t_{ES\_m}$). When the variant state \emph{E} is deleted, all entry and exit transitions linked to \emph{E} are deleted. Simultaneously, new transitions are generated by the composition of the non-optional entry transitions $(t_{AE\_1},\: ...,\: t_{AE\_n})$ with the non-optional exit transitions $(t_{ES\_1},\: ...,\: t_{ES\_m})$.

\begin{figure}[!h] 
\begin{center}
 \includegraphics [width=160mm, height=27mm]{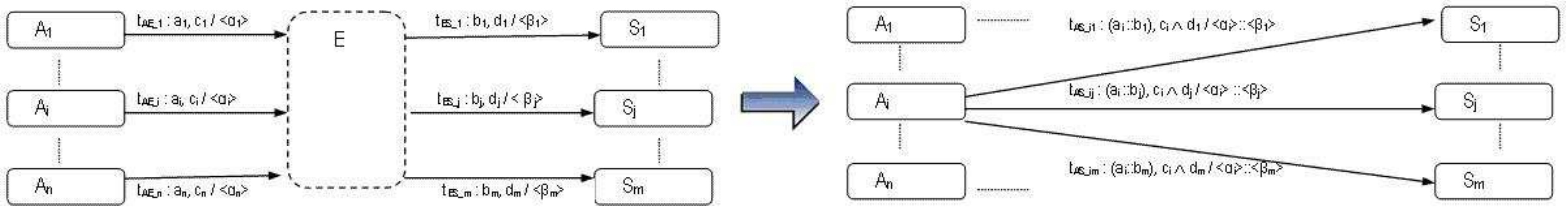}
 \end{center}
\caption{Resulting SC$^*$ after the deletion of the optional state.}
\label{fig:SCGeneralizadodelCaso1}
\end{figure}

The composition of two transitions $t_1 = (t_1,\:so_1,\: e_1,\: c_1,\: \alpha_1,\: sd_1,\: ht_1,\: no\_optional)$ and $t_2 = (t_2,\: so_2,$\\$e_2,\: c_2,\: \alpha_2,\: sd_2,\: ht_2,\: no\_optional)$ define a new transition as follows: $comp(t1,\:t2)=(t_{12},\: so_1,\:e_1::e_2,$\\$c_1\wedge c_2,\: \alpha_1::a_2,\: sd_2,\: ht_2,\: no\_optional)$, where $::$ is the sequential composition of events and actions, and $\wedge$ the conjunction of conditions. Both operations must be associative in order to make the instantiation method deterministic.

Let $sc = (S^*,\: TR^*,\: \Pi^*,\: A^*)$ be SC$^*$ (in future, we will omit the components $\Pi^*$ y $A^*$), we define the set of entry and exit non-optional transitions pairs of a state \emph{E} of $sc$ as follows: 

$T_{e\_s}(E) = \{(t_e, t_s) \in TR^* \times TR^* | \:\:target(t_e)=E \wedge \:source(t_s)=E \wedge \: TransType(t_e)= TransType(t_s)=no\_optional\}$.

The result of eliminating an optional Simple-state \emph{E} of $sc$ corresponds to the following SC$^*$:\\
$Delete\_simple\_state (E, (S^*, TR^*)) = (S^*- \{E\},\;\; TR^* \cup \{comp(t_e,\:t_s) |\;\; (t_e, t_s) \in T_{e\_s}(E) \} -  \{ t \in TR^* \:|\: t \in Domain(T_{e\_s}(E))\ \vee \ t \in Range(T{e\_s}(E)) \} - \{ t \in TR^* |\: TransType(t)=optional \wedge (source(t)=E \vee target(t)=E ) \})$
We call this rule $Delete\_simple\_state(E,\:(S^*,\:TR^*))$. Figure \ref{fig:SCGeneralizadodelCaso1} shows the result of their application.

Case 1.2. \emph{\textbf{When a Or-state is deleted}}

If an Or-state $s = [E,\: (s_1,\:...,\: s_k),\: l,\: T]$ is deleted, then their entry and exit optional transitions are deleted, while the mandatory transitions are composed using the following rebuilding method.

The proposal consists of applying the previous transition composition method of case 1.1 on \emph{E}, considering certain conditions and affectations to SC$^*$. Let $sc =(S^*,\: TR^*)$ be a SC$^*$, we previously define the set of all the entry and exit non-optional transition pairs of an Or-state \emph{E} of \emph{sc$^*$} as follows:\\
$T_{e\_s}(E) = \{ (t_e, t_s) \in TR^* \times TR^* |\: target(t_e)\in sub\_states(E) \wedge \:\: source(t_s)\in sub\_states(E) \wedge \: TransType(t_e)\\= \:TransType(t_s) \: = \: no\_optional\}$

We establish that each entry transition to \emph{E} is composed with one exit transition if the source state of the exit transition is reached from the target state of the entry transition. We define \emph{Reachable(E, A)} as the set of reachable substates of \emph{E} from the substate \emph{A}. Formally, we define $TComp_{e\_s}(E) = \{ (t_e, t_s)\in T{e\_s}(E) | \: source(t_s)\in Reachable(E, target(t_e)) \}$ as the set of transition pairs that must be composed by means of case 1.1, previous modification of these transitions as is indicated as follows. For each entry transition $t_e\in Domain(TComp_{e\_s}(E))$ its target state is now \emph{E}, i.e., $target(t_e)=E$. Also, for each exit transition $t_s\in Range(TComp_{e\_s}(E))$, $source(t_s)=E$. The result of eliminating the optional Or-state \emph{E} of \emph{sc$^*$} corresponds to the SC$^*$ following:

$Delete\_Or\_state (E,\:(S^*,\:TR^*)) = (S^* - (\{E\} \cup sub\_states(E)),\:\: TR^* \ \cup \ \{comp(change\_target(t_e,E),\\ change\_source(t_s ,E)) |\: (t_e, t_s)\in TComp_{e\_s}(E) \} - \{ t\in TR^* | source(t)\in (\{E\} \cup sub\_states(E)) \vee  target(t)\\ \in (\{E\} \cup sub\_states(E)) \})$

\emph{change\_target(t,E)} change the target of transition \emph{t}, such that \emph{target(t)=E}. Likewise, \emph{change\_source(t, E)} change the source of transition \emph{t}, such that \emph{source(t)=E}.
We call this rule \emph{Delete\_Or\_state(E, (S$^*$, TR$^*$))}.

Case 1.3. \emph{\textbf{When an And-state is deleted}}

If an optional And-state \emph{E} is deleted, then their entry and exit optional transitions are deleted, while the mandatory transitions are composed using a similar rebuilding method to case 1.2.

Two possible relations of dependency or synchronization between parallel states exist. One of them refers to the occurrence of an event that produces the trigger of two or more transitions belonging to each one of the parallel substates. The second relation corresponds to using conditions of type `\emph{in E}' (see case 2.2). The latter forces to redefine the concept of reachability, since it is not valid to apply the previous definition of reachability in a way independent in each one of the orthogonal states.

Let \emph{E} be an And-state with \emph{n} orthogonal states. We define now $Reachable(E , (E_1, E_2 , ..., E_n ))$ as the set of n-tuples of reachable states from $(E_1, E_2 ,..., E_n )$. In this way, maintaining the definition of $T_{e\_s}(E)$ of the previous case and redefining $TComp_{e\_s}(E)$, it is possible to solve the method of transition elimination and composition (\emph{Delete\_And\_state}) in an analogous form to case 1.2.

$TComp_{e\_s} (E) = \{ (t_e, t_s)\in T_{e\_s}(E) | (\exists \  n$-$tupla\_init, n$-$tuple\_end \in (S^* \times S^* \times ... \times S^*)_n | \;\; n$-$tupla\_end \in Reachable(E,\: n$-$tupla\_init) \wedge (\exists i,j 1\leq i,j\leq n | n$-$tupla\_end[i] = source(t_s) \wedge n$-$tupla\_init[j] = target(t_e) \wedge cond(t_s) )) \}$, 
being \emph{n} the amount of orthogonal states.
We call to this rule \emph{Delete\_And\_state (E, (S$^*$, TR$^*$))}.\\

Case 2. \emph{\textbf{Consequences of the elimination of a state}}

Some situations can appear as a consequence of applying the  cases described previously, which must be considered in order to reestablish the SC. These situations are analyzed in the following cases.

Case 2.1. \emph{\textbf{When an initial state is deleted}}

If an initial state of a state $E = [s,\: (s_1,\:...,\: s_k),\: l,\: T]$ is eliminated, then anyone of their successors that belongs to \emph{E} becomes the new initial substate, i.e. if $\exists \: t, t=(t_{name},\: s_1,\: e,\: c,\: \alpha,\: s_{new\_initial},\: ht) \in T | s_{new\_initial} \in (s_2,\:...,\: s_k)$ then $initial'(E) = s_{new\_initial}$.
We call this rule \emph{Delete\_Initial\_state(E, (S$^*$, TR$^*$))}.

Case 2.2.  \emph{\textbf{When some substate in a parallel decomposition is deleted}}

The conditions of transitions in a parallel decomposition of type `\emph{in E}' are deleted when the state \emph{E} is eliminated via some FM configuration, i.e. if \emph{E} is deleted and $t = (t_{name},\: s_o,\: e,\: c,\: \alpha,\: s_d,\: ht) \in TR |\: `in\: E$'$ \in c$, then $t' = (t_{name},\: s_o,\: e,\: delete\_expresion\_in\_condition(`in E$'$,c),\: \alpha,\: s_d, \:ht)$, where the function \emph{delete\_expresion\_in\_condition(expr, c)} eliminates the logic subexpression of \emph{c}.
We call this rule \emph{Delete\_condition(E, (S$^*$, TR$^*$))}.\\


Case 3. \emph{\textbf{When a transition is deleted}}

If a transition \emph{t} of a SC$^*$ disappears, it does not produce alterations in the SC$^*$.
We call this rule \emph{Delete\_transition(t, (S$^*$, TR$^*$))}.\\

Case 4. \emph{\textbf{Changing the optional elements to non\_optional elements}}

This rule will be used in the reduction strategy, in order to make part of the final product all optional elements remainders. That is to say, $\forall  s \in SOp^*, StateType(s') = non\_optional$. In a same way, $\forall t \in TR^*, TransType(t') = non\_optional$.
We call this rule \emph{Changing\_optional\_to\_non\_optional(S$^*$, TR$^*$)}.

\subsection{Instantiation Method}
\label{Instantiation Method}

Given a FM and its configuration, we will call \emph{NSF} the set of \emph{non-selected functionalities} of the model as consequence of the configuration. Formally, for the FM $fm = (Funcs,\: f_0,\: Mand,\: Opt,\: Alt,\: Or$-$rel)$ and a configuration $conf_{fm} = (F,\: R)$ of \emph{fm}, \emph{NSF = Funcs $-$ F}. We define also the set of \emph{non-selected components(NSC)} by the configuration $conf_{fm}$ of SC$^*$ \emph{sc}, which will not be part of the resulting SC, as follows: $NSC(conf_{fm},\: sc) = \{ x \in VarElem | \exists f_i \in NSF: x \in Imp(f_i) \wedge \neg \exists f_i' \in F: f_i' \neq f_i \wedge x \in Imp(f_i') \}$, with \emph{VarElem} e \emph{Imp} defined for \emph{sc} according to section 4.4. This is, the states and transitions that do not implement selected functionalities by a configuration will be excluded, through the rules, from the behavior of resulting SC.

\subsubsection{Application Strategy}
The rules should be executed in a certain order for the result to be deterministic. We will assume that rules are organized in rule sets, which are then organized in a sequence of rule sets in which each rule set can be considered as a layer \cite{JK06}. Within a rule set, rules maybe applied in a non-deterministic order. Syntactically we express layers of rule sets as  shown in the following example: given three rules \emph{p1}, \emph{p2}, and \emph{p3}, $(\{p1, p2\}, p3\downarrow)$ specifies two layers. The first one contains \emph{p1}, \emph{p2}, and the second contains \emph{p3}. This means that first any rule in the first layer is applied, and then the one in the second layer is applied. The symbol $p\downarrow$  denotes that the rule $p\downarrow$ is iterated until it cannot be applied anymore.

The order of application of the rules will determine the order of selection of the elements in \emph{NSC (conf$_{fm}$, sc)}. Let $sc \in SC^*$, $fmc \in FMConf$, and $E,t \in NSC(fmc,\: sc)$. We establish the following application order of rules:
\begin{center}
$S_R  = ( \{Delete\_Initial\_state(E,\: sc),\: Delete\_condition(E,\: sc)\}\downarrow ;$\\$ \{Delete\_simple\_state(E,\: sc),\: Delete\_ Or\_state(E,\: sc),\: Delete\_And\_state(E,\: sc)\}\downarrow ;$\\$ Delete\_transition(t,\: sc)\downarrow ;\: Changing\_optional\_to\_non\_optional(sc))$
\end{center}

The strategy should be applied only once. It is important to consider that the application of a rule produces the elimination of elements in $NSC(fmc,\: sc)$. The strategy is completed when $NSC(fmc,\: sc)= \emptyset$.
Note that the implementation of the first two rules does not eliminate states, but only change initial states and deleting conditions like `in E'. Later the affected states will be eliminated by the successor rules. The rule \emph{Changing\_optional\_to\_non\_optional(sc)} convert optional elements, which remain in \emph{sc}, to non\_optional elements.

\subsubsection{Termination and Deterministic Implementation}

Termination and confluence problems can occur whenever a rule-based approach is used. As termination and confluence are fundamental for the correctness of a model transformation, systematically validating these properties is a main prerequisite for successful practical applications, such as transformations in the MDA context, for example.

In our context, the property of termination is clear, given that rules are applied on elements of \emph{NSC(fmc, sc)}, and these are eliminated in each application until $NSC(fmc,\: sc)=\emptyset$ .
Therefore the strategy is terminating because the rules decrease components of the statecharts (either states or transitions).

In general, confluence means that, when there is a choice of applying different rules, the choice does not affect the result.
The role of the first two rules is to avoid ill-formed statecharts. The main rule is \emph{Delete\_simple\_state(E, sc)}, since \emph{Delete\_Or\_state(E, sc)} and \emph{Delete\_And\_state(E, sc)} are based on this. Note that $Delete\_simple\_state(E, sc)$ does not produce ill-formed statecharts (see case 1.1).

The application of a rule \emph{r} on a SC$^*$ is denoted by $sc_i^* \stackrel{r}{\rightarrow} sc_{i+1}^*$, and $sc_i^* \stackrel{(r_1,\: r_2)}{\rightarrow} sc_{i+2}^*$ denotes the application of $r_1$ y $r_2$ in that order.

\begin{lem}(Applying two rules.)
Let $r_1$, $r_2 \in \{Delete\_simple\_state(E,\: sc),\; Delete\_Or\_state(E,\: sc),\; Delete\\\_And\_state(E,\: sc)\}$ be two rules to be applied to states $E1$ and $E2 \in NSC(fmc,\: sc)$, respectively.
\begin{center}
$sc_i^* \stackrel{(r_1,\: r_2)}{\rightarrow} sc_j^* \; \wedge \; sc_i^* \stackrel{(r_2,\: r_1)}{\rightarrow} sc_t^* \;\;\; \Rightarrow \;\;\; sc_j^* = sc_t^*$.
\end{center}
\end{lem}

\begin{proof}
Let $E_1$ and $E_2 \in  NSC(fmc,\: sc)$  be the states to eliminate. We consider both independent if there is no transition which connects, in this case it is clear that the elimination in any order produces the same SC. 
The opposite case is discussed.
Let $t_{E_1} = (t_{E_1},\: so_1,\: a_1,\: c_1,\: \alpha_1,\: E_1,\: ht_1,\: non\_optional)$, $t_{E_1\_E_2} = (t_{E_1\_E_2},\: E_1,\: a_{12},\: c_{12},\: \alpha_{12},\: E_{12},\: ht_{12}, non\_optional)$ and $t_{E_2} = (t_{E_2},\: E_2 ,\: a_2,\: c_2,\: \alpha_2,\: sd_2,\: ht_2,\: non\_optional)$ be entry and exit transitions of $E_1$ and $E_2$.
On the one hand, we assume that the rule $Delete\_simple\_state(E_1,\\sc)$ is applied before  rule \emph{Delete\_simple\_state($E_2$, sc)}, i.e. $t_{E_1}$ is composed with $t_{E_1\_E_2}$ and after compose with $t_{E_2}$.

$comp(comp(t_{E_1},\: t_{E_1\_E_2}), t_{E_2}) = t_{E_1\_(E_1\_E_2)\_E_2}  = (t_{E_1\_(E_1\_E_2)\_E_2},\: so_1,\: (a_1 ::  a_{12}) :: a_2,\:  (c_1 \wedge c_{12}) \wedge c_2,\: (\alpha_1 :: \alpha_{12}) :: \alpha_2 ,\: sd_2 ,\: ht_2,\: non\_optional)$.



On the other hand, if we assume now the inverse application of both rules, we can see that the resulting composition has the same result as the previous case.
Given that the operators $\wedge$ and $::$ are associative: $comp(comp(t_{E_1},\: t_{E_1\_E_2}),\: t_{E_2}) = comp(t_{E_1},\: comp(t_{E_1\_E_2},\: t_{E_2}))$. 

The rules \emph{Delete\_Or\_state} and \emph{Delete\_And\_state} are based on the rule \emph{Delete\_simple\_state}, therefore we may consider $E_1$ and $E_2$ (see above) as compound states.
\end{proof}

Let SC$^*$, $R$ y $S_R$ be a rule-based rewriting system on Statecharts (RS-SC), where $R$ is the set of rules defined in section \ref{Cases and Rebuilding Rules} and $S_R$ the implementation strategy on the rules of $R$. Let $ValSeq\_S_R$ be the set of all valid sequences of applications.

An RS-SC execution starting from a SC$^*$ ($sc^*$) is a sequence $s$ of rule applications of $R$ such that $s \in ValSeq\_S_R$. We denote it by $sc^* \stackrel{s}{\rightarrow} sc$, where $sc \in$ SC is a concrete statechart.

\begin{thm}(Confluence.) Two distinct executions of RS-SC generate the same statechart.
\begin{center}
$sc^* \stackrel{s_1}{\rightarrow} sc \wedge sc^* \stackrel{s_2}{\rightarrow} sc' \;\;\; \Rightarrow \;\;\; sc =  sc'$.
\end{center}
\end{thm}

\begin{proof}
The proof is given by the following:\\
Both executions, $s_1$ and $s_2$, have the same length, since the number of applications of rules is equal to $\#NSC$. Moreover, the set of rules of sequence $s_1$ is equal to set of sequence $s_2$, therefore $s_1$ is a permutation of $s_2$.
Since the rules $Delete\_Initial\_state$, $Delete\_condition$ and $Changing\_optional\_to\_non\_op\-tional$ of $S_R$ are clearly deterministic and do not remove states or transitions, we focus only on the subsequences of $s_1$ ($sub\_s_1$) and $s_2$ ($sub\_s_2$)	which include rules $Delete\_simple\_state$, $Delete\_Or\_state$ and $Delete\_And\_state$.


Since $sub\_s_2$ is a permutation of $sub\_s_1$, the proof follows immediately from Lemma 1 by swap rules in any of sequences. 
\end{proof}

\subsection{Instantiation of the Case study: MPs}
\label{Instantiation of the Case study: MPs}

The FM of figure \ref{fig:fm-del-ejemplo} can be configured to characterize different MPs, according to the specification of the case study of section \ref{Case study: MPs}. Next, we present a configuration of a MP of figure \ref{fig:SC*deunTM} and we proceed towards obtaining the corresponding SC (the MP wanted), according to the application order of rules defined in section \ref{Instantiation Method}. 

A MP with neither the support for the management of polyphonic sounds nor the capacity of alerting the user when new messages enter in the incoming mailbox, is defined by the configuration $conf_{fm} = (F,\: R)$ of the FM of figure \ref{fig:fm-del-ejemplo}, where: 

{\small $F = \{\letraFunc{MP,\: Display,\: Contacts,\: MessagesAdm,\: Multimedia,\: Images,\: Videos,\: Quick}$-$ \letraFunc{Marking,Ringer\_in\_func}\\ \letraFunc{tions\}}$}, and

{\small $R = \{\letraFunc{	(MP,\: \{Multimedia\}),\:\: (MP,\: \{MessagesAdm\}),\:\: (MP,\: \{Quick}$-$ \letraFunc{Marking\}),\:\: (MP,\: \{Display\}),\:(MP,\: \{Con}\\ \letraFunc{tacts\}),\:\: (MP,\: \{Ringer\_in\_ functions\}),\:\: (Multimedia,\: \{Images\}),\: (Multimedia,\: \{Videos\})\}}$.}

Taking the previous configuration and the function \emph{Imp} described in example 1 of section \ref{Relation between FMs and SCs*}, the sets $NSF$ and $NSC(conf_{fm},\: sc)$ are defined as follows: 

{\small $NSF = \{\letraFunc{Polyphonic Sounds, AlarmNewMessages\}}$}, and

{\small $NSC(conf_{fm},\: sc) = \letraEstado{\{SelectPolSound, TRightMultimediaType-SelectPolSound, ToChoosePolSou}
\\ \letraEstado{nd, TRightSoundType-ToChoosePolSound, TLeftToChoosePolSound-SoundType,
 TRightToChoosePo}\\ \letraEstado{lSound-PhoneFuncionality,} \letraEstado{ MessagesState, TMessage-MainDisplay-IncomingMess\}}$.} 

In accordance with the strategy defined in section \ref{Instantiation Method}, must apply the rules in the following order:

{\small $Delete\_condition (MessagesState, sc) \rightarrow$ $Delete\_simple\_state(SelectPolSound, sc) \rightarrow$
$Delete\_simple\_state(To\\ChoosePolSound, sc) \rightarrow$ $Delete\_transition(TRightMultimedia$-$Type$-$SelectPolSound , sc) \rightarrow$ $Delete\_transition(\\TRightSoundType$-$ToChoosePolSound, sc) \rightarrow$ $Delete\_transition(TLeftToChoosePolSound$-$SoundType, sc) \rightarrow $ \\$Delete\_transition(TRightToChoosePolSound$-$PhoneFuncionality, sc) \rightarrow$
$Delete\_transition(TMessage$-$MainDis\\play$-$Incoming Mess, sc) \rightarrow$
$Changing\_optional\_to\_non\_optional(sc).$}

The resulting SC is observed in figure \ref{fig:SCresulting}.

\begin{figure}[!h]
\begin{center}
\includegraphics [width=150mm, height=135mm]{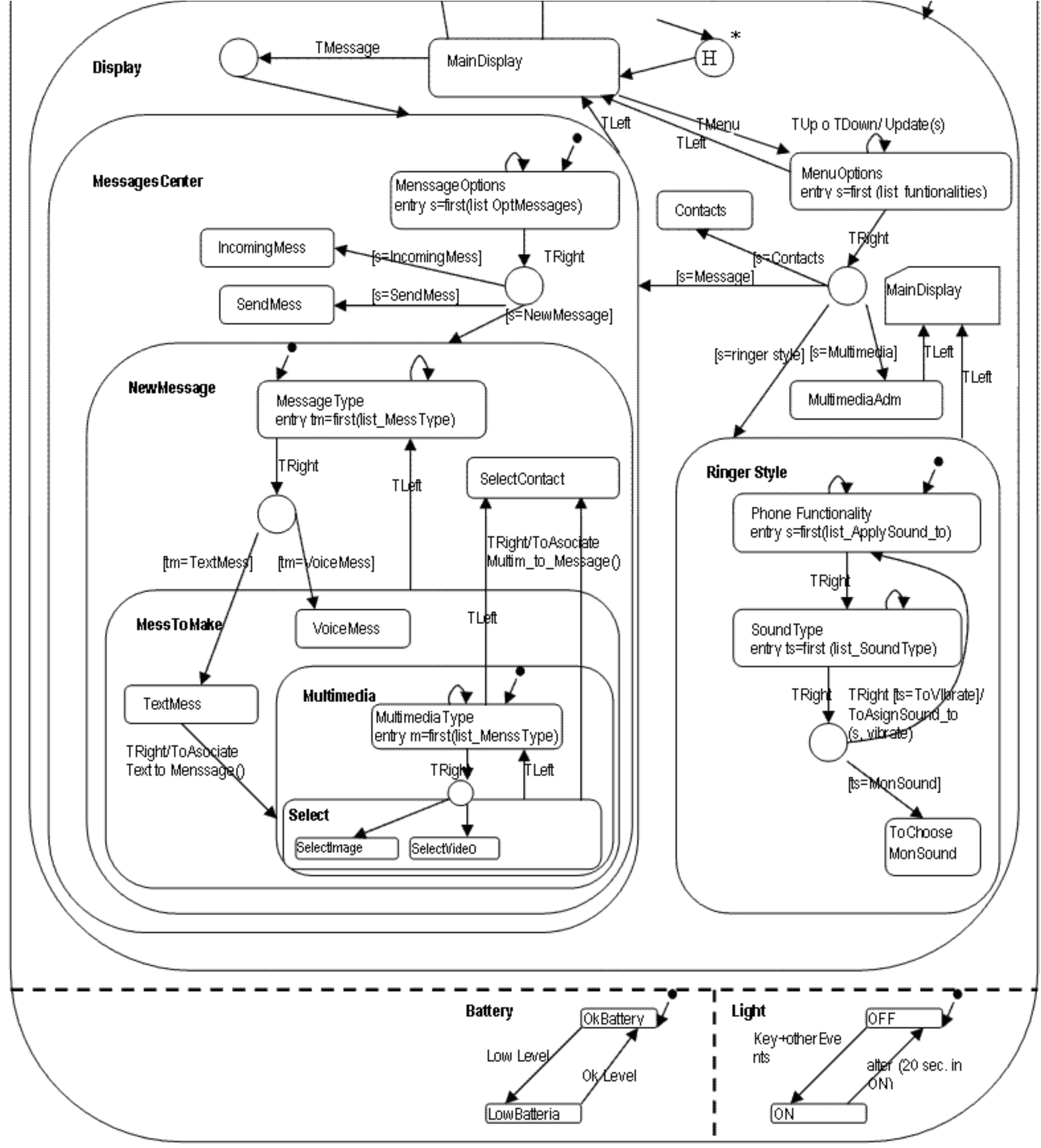}
\end{center}
\caption{SC for a MP without polyphonic sounds and without alert of new messages.}
\label{fig:SCresulting}
\end{figure}

\section{Related Works}
\label{RelatedWorks}
A variety of existing approaches propose the incorporation of variability in software systems and in particular on PLs. One of these is the one designed by Jacobson \cite{JGJ97}, whose weaknesses have been analyzed by several authors.
Von der Maßen \cite{vdML02} proposes the use of a graphical notation with new relations. John and Muthig \cite{JM02} suggest the application of use case templates, although they do not distinguish between optional, alternative or obligatory variants. However, Halman and Pohl propose in \cite{LGL07} to make use of UML 2 package merge, based on \cite{ZDD06}, as a tool for the variability representation and configuration in PLs. As opposed to the mentioned proposals, our solution is centered in one behavior specification model of a PL, as are the SCs, introducing a clearly defined formal sustenance.

Also, to define PLs and to characterize their different products we use FMs, which admit a formal definition and allow us to configure the functional characteristics of a line. An alternative approach to this paper is developed in parallel by other members of the research project in which this work is subsumed. In \cite{SV08} the authors, in a formal framework, define functions that associate SCs (not components of SCs, as in our case) to functionalities of a FM, and analyze forms of combination between different SCs which specify possible variants of a PL. Whereas under our method the behavior of a product into a PL is obtained basically by a selection process, in \cite{SV08} the focus is oriented towards a process of SC combinations.


\section{Conclusion and Further Work}
\label{ConclusionandFurtherWorks}
Most of the techniques of Model Driven Development make use of UML. In particular, the SCs of UML constitute a mechanism for specifying systems behavior by means of a graphical representation. In this work, we presented an extension of UML SCs by incorporating variability in their essential components to specify PLs. The variability is introduced in the SCs distinguishing optional and non-optional states as well as optional and non-optional transitions. A PL is specified with a SC$^*$, a FM, and a formal relation (an implementation function) that binds both models. Using FMs to describe the common and variant functionalities and applying a rule-based instantiation method, concrete SCs corresponding to different PLs can be obtained. The approach defines the transformation strategy from extended SCs to standard UML SCs. We develop partial examples of a case study based on mobile phone technology, whose full version is not included in this article due to space restrictions.

Given the fact that UML and SCs have become very successful languages for analysis and design in the very short run, we are confident that the results of this work can be successfully applied to the real problems of the software industry. It is a timely contribution to an authentic and actual problem. 

As part of our plans for future work, we are interested in an extension of SCs which allows us to completely cover the UML 2.0 SCs and analyze variabilities, not only the ones considered in this paper, but in all of their components. Also, we will make an attempt to provide a formal semantics for the extension. This semantics is an essential preliminary step towards both the automatic code generation and the validation of complex software systems. 

\bibliographystyle{eptcs} 

\end{document}